\documentclass[conference]{IEEEtran}
\usepackage[T1]{fontenc}
\usepackage[utf8]{inputenc}
\usepackage{microtype}
\usepackage{graphicx}
\usepackage{xcolor}
\usepackage{booktabs}
\usepackage{array}
\usepackage{amsmath}
\usepackage{listings}
\usepackage[numbers,sort]{natbib}
\usepackage{caption}
\usepackage{enumitem}
\usepackage{url}
\usepackage[hidelinks]{hyperref}
\usepackage{tikz}
\usetikzlibrary{arrows.meta,positioning,calc}
\usepackage{pgfplots}
\pgfplotsset{compat=1.18}
\usepgfplotslibrary{groupplots}

\definecolor{assaysurface}{HTML}{F7F7F5}
\definecolor{assayblue}{HTML}{2A78D6}
\definecolor{assayorange}{HTML}{EB6834}
\definecolor{assayaqua}{HTML}{1BAF7A}
\definecolor{assaymuted}{HTML}{52514E}
\definecolor{assaygood}{HTML}{1E7A3C}
\definecolor{assaybad}{HTML}{C62828}
\definecolor{assaywarn}{HTML}{B8860B}

\tikzset{
  mod/.style={circle, draw, minimum size=6mm, inner sep=0pt, font=\scriptsize\bfseries, line width=0.7pt},
  fresh/.style={mod, fill=assayblue!12, draw=assayblue, text=black},
  stalemod/.style={mod, fill=assayorange!18, draw=assayorange, text=black},
  changed/.style={mod, fill=assayorange, draw=assayorange, text=white},
  dep/.style={-{Stealth[length=1.8mm]}, draw=assaymuted, line width=0.6pt},
  tag/.style={rectangle, rounded corners=1.5pt, draw, inner sep=2.2pt, font=\tiny, align=center, line width=0.5pt},
  tagok/.style={tag, fill=assaygood!7, draw=assaygood, text=assaygood},
  tagbad/.style={tag, fill=assaybad!7, draw=assaybad, text=assaybad},
  state/.style={rectangle, rounded corners=2pt, draw, minimum width=16mm, minimum height=5.6mm, inner sep=2pt, font=\scriptsize\bfseries, align=center, line width=0.7pt},
  tr/.style={-{Stealth[length=1.8mm]}, line width=0.7pt},
  trlab/.style={font=\tiny, text=assaymuted, inner sep=1pt, fill=white},
}

\newtheorem{proposition}{Proposition}
\newtheorem{definition}{Definition}
\newenvironment{proof}[1][Proof]{\par\noindent\textit{#1:}\ \ignorespaces}{\hfill\rule{0.45em}{0.45em}\par}
\newcommand{\assay}{\textsc{Assay}}
\newcommand{\cone}{\operatorname{cone}}
\newcommand{\BR}{\operatorname{BR}}
\newcommand{\hash}{\mathsf{h}}
\newcommand{\chash}{\hat{\mathsf{h}}}
\newcommand{\Mods}{\mathcal{M}}
\newcommand{\Deps}{\mathcal{D}}
\setlist{nosep, leftmargin=1.4em}
\makeatletter
\renewenvironment{abstract}{\normalfont
  \@IEEEtweakunitybaselinestretch{1}
  \bfseries\itshape\footnotesize
  \noindent Abstract:\ \ignorespaces}{\par}
\renewenvironment{IEEEkeywords}{\normalfont
  \@IEEEtweakunitybaselinestretch{1}
  \bfseries\itshape\footnotesize
  \noindent Index Terms:\ \ignorespaces}{\par}
\makeatother

\newcommand{\EstimatorName}{tiktoken/cl100k\_base}

\newcommand{\EOneExpressModules}{35}

\newcommand{\EOneFlaskModules}{16}

\newcommand{\EOneGinModules}{7}

\newcommand{\EOneAxumModules}{90}

\newcommand{\EOneFastapiIndexModulesOnly}{51,842}

\newcommand{\EOneFastapiModules}{195}

\newcommand{\EOneAssayModules}{5}

\newcommand{\EOneMinXFull}{216}
\newcommand{\EOneMaxXFull}{1225}
\newcommand{\EOneMinXExpl}{14}
\newcommand{\EOneMaxXExpl}{114}
\newcommand{\EOneMinBrief}{592}
\newcommand{\EOneMaxBrief}{635}

\newcommand{\ETwoFastapiCold}{508}
\newcommand{\ETwoFastapiWarmEdit}{106}

\newcommand{\ETwoFastapiFiles}{1142}

\newcommand{\ETwoMaxSpeedupNoop}{5}
\newcommand{\ETwoMinSpeedupNoop}{2}

\newcommand{\EThreeTrials}{100}
\newcommand{\EThreeMinConeCost}{7.9\%}
\newcommand{\EThreeMaxConeCost}{81.9\%}
\newcommand{\EThreeMinSelfMissed}{23\%}
\newcommand{\EThreeMaxSelfMissed}{68\%}
\newcommand{\EFourScenarios}{12}
\newcommand{\EFourBlocked}{9}
\newcommand{\EFourPassed}{3}
\newcommand{\EFourAllExpected}{all}

\newcommand{\EFiveAxumMedian}{1.1\%}
\newcommand{\EFiveAxumPninety}{91.1\%}

\newcommand{\EFiveAxumTopModule}{axum-macros/src}
\newcommand{\EFiveAxumTopBr}{84}

\newcommand{\EFiveFastapiMedian}{0.5\%}

\newcommand{\EFiveMinSingle}{0\%}
\newcommand{\EFiveMaxSingle}{81\%}
\newcommand{\LocCore}{2,313}

\newcommand{\NumTests}{30}

\begin{document}

\title{Assay: Claims That Decay With the Code\\
\large Content-Addressed Evidence Graphs for Accountable AI-Assisted Software Delivery}

\author{\IEEEauthorblockN{Om Shankar Tiwari}
\IEEEauthorblockA{Applied AI Technical Lead, Google\\California, United States\\\url{https://orcid.org/0009-0001-2247-0383}}
\and
\IEEEauthorblockN{Tangi Vass}
\IEEEauthorblockA{Co-Founder \& CPTO, OMNI3ai\\Lyon, France\\\url{https://orcid.org/0009-0009-0414-6190}}
\and
\IEEEauthorblockN{Gagan Deep Singh}
\IEEEauthorblockA{Founder, Glinr Studios and theSVG.org\\Texas, United States\\\url{https://orcid.org/0009-0002-0579-6514}}}

\maketitle
\setlength{\abovedisplayskip}{4pt plus 1pt}\setlength{\belowdisplayskip}{4pt plus 1pt}
\setlength{\abovedisplayshortskip}{2pt}\setlength{\belowdisplayshortskip}{2pt}

\begin{abstract}
AI coding agents fail in two coupled ways. They spend most of their context window rediscovering where things live, and they assert success without evidence when the work gets hard. Repository indexes address the first with cheap context, and orchestration frameworks with adversarial review address the second with accountability. Both describe the same object, the structure of the codebase, at two timescales: what is true of the code now, and what was verified to be true, at which revision, by whom. \assay{} makes that observation operational. Every claim an agent makes (\emph{tests pass}, \emph{no secrets}, \emph{behavior preserved}) is bound to the Merkle hash of the dependency cone of the code it covers, so the claim is stale exactly when that code or anything it depends on changes. We show the binding is sound and minimal, and that the blast radius of a change is precisely the set of claims it invalidates. On the graph we place a risk-proportional evidence obligation, a bounded review protocol with separation of duties, and a merge gate that consults no model: coverage, freshness, signatures, exit codes, plausibility, evidence monotonicity (the mechanical form of ``do not delete the failing test''), and review status. \assay{} is a dependency-free Python tool with an MCP server. On five public repositories a 600-token brief costs \EOneMinXExpl$\times$ to \EOneMaxXExpl$\times$ less than an exploration proxy, warm rebuilds are up to \ETwoMaxSpeedupNoop$\times$ faster than cold ones, cone binding re-verifies \EThreeMinConeCost{} to \EThreeMaxConeCost{} of claims where repository binding re-verifies all of them while per-module binding misses \EThreeMinSelfMissed{} to \EThreeMaxSelfMissed{} of required invalidations, and the gate blocks \EFourBlocked{} of \EFourBlocked{} scripted adversarial behaviours while admitting the honest ones. Every number in this paper is generated by the released scripts.
\end{abstract}


\section{Introduction}

An agent that edits code has to answer two questions before it can be trusted: where does this change go, and how do we know it worked? Both are questions about the structure of the repository. The first is about the structure now. The second is about the structure at the moment a check ran, and whether it has drifted since.

Current tooling answers them with two disconnected artifacts. Context tools produce maps. Aider's repository map~\citep{aider}, packers such as Repomix~\citep{repomix}, and committed indexes such as Stacklit~\citep{stacklit} compress a repository into a few hundred to a few thousand tokens so an agent stops reading files to learn what they contain. Governance frameworks produce process state. Liza~\citep{liza} is the most disciplined example we know. It pairs every doer with an independent reviewer, makes verdicts binding, and moves the invariants that matter most (no self-approval, no forbidden transitions, no merge without review) out of the prompt and into compiled code, on the grounds that instructions bend under pressure and code does not.

Each lineage stops where the other starts. A map knows that module $A$ depends on module $B$, but not that the tests covering $A$ last ran before $B$ changed. A task board knows a reviewer approved task 17 at commit \texttt{9b3c}, but not which modules the approval is about. A later change to an unrelated module and a change to something $A$ depends on look identical to it: both are ``after the approval''. That forces a bad choice. Either every accepted verdict stays valid until someone notices otherwise, or every change invalidates everything and review starts over.

A claim about code is a claim about content, and content is content-addressable. Key a claim by the hash of the code it is about and ``has this gone stale?'' becomes a hash comparison rather than a judgment. The subtlety is which hash. Hashing only the named module is unsound, because a test result for $A$ depends on $B$. Hashing the whole repository is sound but useless, because every edit invalidates every claim. The right key is the hash of the dependency cone of the subject, computed as a Merkle hash over the condensation of the dependency graph. Under this key staleness is sound and minimal, and the blast radius of a change, the modules whose cone hash moved, is by construction the set of claims it invalidated. Context and accountability become one computation over one graph.

This paper contributes:
\begin{itemize}
  \item \textbf{A model.} Content-addressed modules, cone hashes over the condensation DAG, claims bound to cone hashes, and propositions establishing soundness, minimality, and the identity between blast radius and staleness frontier. A ledger of such claims is a build cache for assertions in the sense of Mokhov et al.~\citep{mokhov2018build}.
  \item \textbf{A protocol.} Risk-proportional evidence obligations, a bounded review cycle with separation of duties adapted from Liza and applied per claim, and a merge gate with eight mechanical checks including evidence monotonicity.
  \item \textbf{A system.} \assay{}, a zero-dependency implementation with an incremental Merkle-skipped indexer for seven language families, a 600-token situation brief, a signed append-only ledger, a CLI, git hooks, and an MCP server.
  \item \textbf{A mechanical evaluation.} Five scripted experiments on five public repositories and on \assay{} itself, with no model in the loop, so every number reproduces to the digit.
\end{itemize}
The implementation, experiments, results, and paper source are public.\footnote{\url{https://github.com/OmShiv/assay-research}}

The paper does not measure how much better language-model agents perform with \assay{} than without. That study needs models in the loop and is future work. The claim here is narrower and more durable: the floor that holds no matter what the model decides can be made cheap and precise, and the structure that provides it also provides the agent's context.

\section{Background}

\subsection{Context tools}

The cost of an agent's first ten minutes in a repository is documented by the tools built to reduce it. Stacklit's comparison reports that packers such as Repomix produce 180k to 400k tokens for repositories of 15k to 37k lines, while its structured index produces a few hundred~\citep{stacklit}. Aider's map ranks symbols by a reference graph and fits them into a budget~\citep{aider}. These tools share a shape: parse, usually with tree-sitter~\citep{treesitter}, extract exports and imports, group into modules, rank by connectivity or recency, emit a compact artifact. Stacklit adds two ideas we keep. The artifact is committed, so every agent and human reads the same map, and regeneration is Merkle-skipped, so a git hook keeps it fresh at negligible cost. It exposes the map through the Model Context Protocol~\citep{mcp}. What such maps do not carry is any notion of what has been verified. When an agent claims the tests pass, the map has nothing to say about whether the claim survives the next commit.

\subsection{Governance frameworks}

Multi-agent frameworks such as MetaGPT~\citep{hong2024metagpt}, ChatDev~\citep{qian2024chatdev}, AutoGen~\citep{wu2023autogen}, and CrewAI~\citep{crewai} assign roles and route messages. Specification-first methods such as Spec Kit~\citep{speckit} and BMAD~\citep{bmad} make the artifacts explicit. Failure taxonomies show why that is not enough. Cemri et al.~\citep{cemri2025mast} catalog specification drift, inter-agent misalignment, and weak verification, and the sycophancy and reward-tampering literature~\citep{sharma2024sycophancy,denison2024subterfuge} documents models that rewrite the test rather than the code when under pressure to appear competent.

Liza~\citep{liza} answers with layers: a behavioral contract naming failure modes and countermeasures, adversarial doer and reviewer pairs whose verdicts bind, a task state machine with forbidden transitions and bounded iteration enforced by a compiled supervisor, per-task worktrees, and a blackboard~\citep{nii1986blackboard,erman1980hearsay} recording every transition. Its review protocol is what we borrow most directly. A finding must be answered by accepting, countering, refuting, or escalating. Repeating an assertion is not a permitted move, and a doer who contests must name the concrete harm the requested change would cause. What Liza's blackboard does not carry is the code. Evidence is a commit SHA and an exit code, and whether a verdict still means anything after later changes is answered procedurally, by re-review after rebase, rather than structurally.

Both lineages need the dependency graph and both compute or approximate it. Only the first keeps it. The second needs it at a different moment: not when the agent reads but when the gate decides. Keying claims by cone hash lets one graph serve both moments.

\section{The Evidence Graph}

\subsection{Modules, hashes, and cones}

A repository state $S$ is a finite set of files with contents, partitioned into modules $\Mods$ (in \assay{}, directories up to a configurable depth). With $H$ a collision-resistant hash, the content hash of a module is the hash of its sorted (path, file hash) pairs:
\[
\hash_S(m) = H\big(\langle (f, H(\text{content}_S(f))) : f \in m \rangle_{\text{sorted}}\big).
\]
Import resolution yields a dependency relation $\Deps_S \subseteq \Mods \times \Mods$, where $(a, b) \in \Deps_S$ means some file in $a$ imports some file in $b$. The relation may contain cycles, as Python packages routinely do. Let $C(m)$ be the strongly connected component of $m$ and $\mathrm{succ}(C)$ the components reachable from $C$ in one step of the condensation DAG.

\begin{definition}[Cone and cone hash]
The dependency cone of $m$ is $\cone_S(m) = \{m\} \cup \{m' : m \to^{*} m'\}$. The cone hash is defined on components in reverse topological order of the condensation DAG:
\begin{multline*}
\chash_S(C) = H\Big(\langle (m, \hash_S(m)) : m \in C \rangle_{\text{sorted}} \;\|\\
\langle \chash_S(C') : C' \in \mathrm{succ}(C) \rangle_{\text{sorted}}\Big),
\quad \chash_S(m) = \chash_S(C(m)).
\end{multline*}
\end{definition}

This is a Merkle tree~\citep{merkle1987} over the condensation DAG with sharing: each component's hash is computed once and referenced by every component above it. Members of a cycle share one cone hash, which is right, since they depend on each other. The root hash of the repository is the hash of all module content hashes and serves as the coarsest binding scope.

\subsection{Claims and bindings}

\begin{definition}[Claim]
A claim is a tuple $c = (p, S_c, \sigma, \beta, e, a, t)$: a predicate $p$ (\emph{tests-pass}, \emph{lint-clean}, \emph{no-secrets}, \emph{behavior-preserved}, \emph{human-approved}, and so on), a non-empty set of subject modules $S_c \subseteq \Mods$, a scope $\sigma \in \{\text{cone}, \text{self}, \text{repo}\}$, bindings $\beta : S_c \to \text{hashes}$ recorded at attestation time, evidence $e$, an attestor $a$ with a role, and a timestamp $t$.
\end{definition}

For scope \emph{cone}, $\beta(s) = \chash_S(s)$ at the state in which the evidence was produced. Evidence is what the tool observed, not what the attestor said: the command that ran, its exit code, a hash of its full output, size and duration, and, when the output is recognisable test-runner output, the number of tests that ran. A claim whose command exited non-zero is recorded as a \emph{failed} claim, never discarded. Failure is evidence too.

\begin{definition}[Freshness]
A cone-scoped claim $c$ is fresh in state $S'$ iff $\beta(s) = \chash_{S'}(s)$ for all $s \in S_c$. Otherwise it is stale.
\end{definition}

\subsection{Staleness is sound and minimal}

Write $\Delta(S, S') = \{m : \hash_S(m) \neq \hash_{S'}(m)\}$ for the modules whose content changed, including modules added or removed. Import edges derive from file contents, so $\Deps_S$ and $\Deps_{S'}$ can differ only at modules in $\Delta$. The propositions hold up to hash collisions.

\begin{proposition}[Soundness]
\label{prop:sound}
If some module in $\bigcup_{s \in S_c} \cone_S(s)$ is in $\Delta(S, S')$, then $c$ is stale in $S'$.
\end{proposition}
\begin{proof}[Proof sketch]
Let $m \in \cone_S(s) \cap \Delta$. Then $\hash_{S'}(m) \neq \hash_S(m)$, so the component containing $m$ has a different own-part and hence a different cone hash. Every component from which $C(m)$ is reachable includes that hash in its successor list, directly or through a chain of components each changed by the same argument, so by induction along reverse topological order every such cone hash changes. $C(s)$ reaches $C(m)$, so $\chash_{S'}(s) \neq \beta(s)$. If the change instead added or removed an edge that alters $\cone(s)$, the file carrying the import changed, so its module lies in $\Delta \cap \cone_S(s)$ or in $\Delta \cap \cone_{S'}(s)$ and the argument applies in that direction.
\end{proof}

\begin{proposition}[Minimality]
\label{prop:minimal}
If no module in $\bigcup_{s \in S_c} \cone_S(s)$ is in $\Delta(S, S')$, then $c$ is fresh in $S'$.
\end{proposition}
\begin{proof}[Proof sketch]
No module in the cone changed content, so no edge originating inside the cone changed, so $\cone_{S'}(s) = \cone_S(s)$ with identical content hashes throughout. The cone hash is a deterministic function of exactly those hashes and that structure.
\end{proof}

Together the propositions say cone binding invalidates a claim if and only if dependency semantics require re-verification. This is the correctness criterion incremental build systems use to decide whether a cached output can be reused~\citep{mokhov2018build}. A claim is a verifying trace whose key is $(p, S_c)$, whose input hashes are the cone hashes, and whose output is the evidence. The gate's freshness check is the build system's ``is this trace still valid'' check, applied to assertions rather than artifacts.

\begin{proposition}[Blast radius is the staleness frontier]
\label{prop:radius}
Let $\BR(\Delta) = \Delta \cup \{m : m \to^{*} d \text{ for some } d \in \Delta\}$. Then $\BR(\Delta) = \{m : \chash_S(m) \neq \chash_{S'}(m)\}$, and the cone-scoped claims that become stale are exactly $\{c : S_c \cap \BR(\Delta) \neq \emptyset\}$.
\end{proposition}
\begin{proof}[Proof sketch]
$m \in \BR(\Delta)$ iff $\cone_S(m) \cap \Delta \neq \emptyset$, which by Propositions~\ref{prop:sound} and~\ref{prop:minimal} holds iff $\chash(m)$ changed.
\end{proof}

Proposition~\ref{prop:radius} is the identity that makes the two lineages one. The number an agent wants before it edits (how far does this reach?) and the number the gate wants after (which evidence is now void?) are the same set, computed by one traversal. Fig.~\ref{fig:cone} shows it on a five-module example.

\begin{figure}[t]
  \centering
  \begin{tikzpicture}[x=0.93cm, y=0.85cm]
    \node[font=\tiny, text=assaymuted] at (2.05,4.05) {before: two claims, both fresh};
    \node[fresh] (core) at (2.0,0.5) {core};
    \node[fresh] (service) at (2.0,1.6) {service};
    \node[fresh] (app) at (2.0,2.7) {app};
    \node[fresh] (tests) at (1.0,3.5) {tests};
    \node[fresh] (util) at (0.4,2.1) {util};
    \draw[dep] (service) -- (core);
    \draw[dep] (app) -- (service);
    \draw[dep] (tests) -- (app);
    \draw[dep] (tests) -- (util);
    \node[tagok, anchor=west] at (2.45,2.7) {c-1: tests-pass\\bound to cone(app)\\fresh};
    \node[tagok, anchor=north] at (0.55,1.55) {c-2: tests-pass\\bound to cone(util)\\fresh};
    \begin{scope}[xshift=4.35cm]
      \node[font=\tiny, text=assaymuted, align=center] at (2.05,4.05) {after editing core: cone(app) changed,\\cone(util) did not};
      \node[changed] (core2) at (2.0,0.5) {core};
      \node[stalemod] (service2) at (2.0,1.6) {service};
      \node[stalemod] (app2) at (2.0,2.7) {app};
      \node[stalemod] (tests2) at (1.0,3.5) {tests};
      \node[fresh] (util2) at (0.4,2.1) {util};
      \draw[dep] (service2) -- (core2);
      \draw[dep] (app2) -- (service2);
      \draw[dep] (tests2) -- (app2);
      \draw[dep] (tests2) -- (util2);
      \node[tagbad, anchor=west] at (2.45,2.7) {c-1: tests-pass\\bound to cone(app)\\stale};
      \node[tagok, anchor=north] at (0.55,1.55) {c-2: tests-pass\\bound to cone(util)\\fresh};
    \end{scope}
    \draw[assaymuted!40, line width=0.4pt] (4.2,0.1) -- (4.2,3.9);
  \end{tikzpicture}
  \caption{Cone hashes and staleness. Arrows point from a module to what it depends on. After an edit to \texttt{core}, every module that reaches \texttt{core} (shaded) has a new cone hash, so the claim on \texttt{app} is stale and the claim on \texttt{util} is not.}
  \label{fig:cone}
\end{figure}
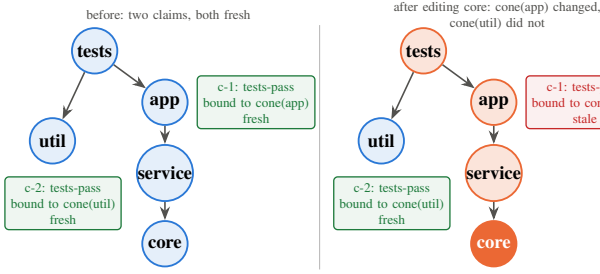

\subsection{Risk and obligation}

Not every change deserves the same review. Each module carries a risk score on absolute scales so it is comparable across repositories:
\[
r(m) = \tfrac{1}{2}\min\!\big(1, \tfrac{\text{commits}_{90d}(m)}{20}\big) + \tfrac{1}{2}\min\!\big(1, \tfrac{\text{fan-in}(m)}{10}\big).
\]
Churn approximates instability and fan-in approximates the cost of being wrong. For a change touching $\Delta$, the obligation level is a function of the radius fraction $|\BR(\Delta)|/|\Mods|$, the radius size, and the maximum risk inside the radius. Level 0 requires passing tests on the changed modules. Level 1 requires them on the whole radius. Level 2 adds a lint claim and a secrets scan on the changed modules and asks for adversarial review. Level 3 adds a human-approved claim. Fraction thresholds apply only once the radius has a minimum absolute size, so a two-module change in a five-module repository is not an escalation. Where a module's language has no test runner, the obligation asks for a build and a typecheck claim instead, so it is never unsatisfiable by construction. The thresholds are policy, stored in the index and overridable per repository. We claim only that they are explicit, monotone in radius and risk, and cheap.

\section{Protocol and Gate}

\subsection{Roles and lifecycle}

Three roles appear in the ledger: doer, reviewer, and human. Any identity may attest. Only reviewers and humans may issue verdicts, and no identity may issue one on its own claim. Checked at append time, this is Liza's ``no self-approval'' invariant moved from task to claim granularity. A human attestation is the escalation authority itself and needs no reviewer, though it can still be refuted if its evidence is contradicted.

Fig.~\ref{fig:lifecycle} shows the derived states. A claim starts \emph{attested}, or \emph{failed} if its evidence exited non-zero. A reviewer may accept it, counter it with a finding, refute it, or escalate it. A countered claim is \emph{contested}, and the attestor must respond before any further verdict, either by accepting the finding and superseding the claim with a stronger one, or by contesting, which requires naming the concrete harm the requested change would cause. Three counters escalate automatically, and only a human verdict resolves an escalation. Refuted and superseded are terminal. Status is never stored. It is derived from the ledger each time it is needed, so the ledger cannot disagree with its own history. Staleness is orthogonal to status and derived from the index: an accepted claim whose bindings drifted is stale, and a stale claim cannot be accepted until re-attested.

\begin{figure}[t]
  \centering
  \begin{tikzpicture}[x=1cm, y=0.8cm]
    \node[state, fill=assayblue!12, draw=assayblue] (att) at (1.0,3.4) {attested};
    \node[state, fill=assayorange!15, draw=assayorange, font=\tiny\bfseries] (con) at (4.25,3.4) {contested\\(up to 3 rounds)};
    \node[state, fill=assaygood!8, draw=assaygood] (acc) at (7.5,3.4) {accepted};
    \node[state, fill=assaywarn!12, draw=assaywarn] (esc) at (4.25,1.9) {escalated};
    \node[state, fill=assaybad!7, draw=assaybad] (ref) at (4.25,0.4) {refuted};
    \node[state, fill=assaymuted!12, draw=assaymuted] (sup) at (7.5,0.4) {superseded};
    \node[state, fill=assaymuted!12, draw=assaymuted] (sta) at (1.0,0.4) {stale};
    \draw[tr, assaygood] (att.north) to[bend left=22] node[trlab, above] {accept: fresh claim, independent reviewer} (acc.north);
    \draw[tr, assayorange] (att) -- node[trlab, above] {counter} (con);
    \draw[tr, assaygood] (con) -- node[trlab, above] {contest, then accept} (acc);
    \draw[tr, assaywarn] (con) -- node[trlab, right] {3 rounds} (esc);
    \draw[tr, assaymuted] (con.south east) -- node[trlab, sloped, below, pos=0.62] {accept-finding + new claim} (sup.north west);
    \draw[tr, assaywarn] (att.south east) -- node[trlab, sloped, above, pos=0.45] {escalate} (esc.west);
    \draw[tr, assaybad] (att.south) to[bend right=18] node[trlab, sloped, below, pos=0.4] {refute} (ref.west);
    \draw[tr, assaybad] (esc) -- node[trlab, right] {human refute} (ref);
    \draw[tr, assaygood] (esc.east) -- node[trlab, sloped, above, pos=0.3] {human accept} (acc.south);
    \draw[tr, assaymuted, dashed] ([xshift=-3mm]att.south) -- node[trlab, left] {drift} ([xshift=-3mm]sta.north);
    \draw[tr, assayblue, dashed] ([xshift=3mm]sta.north) -- node[trlab, right] {re-attest} ([xshift=3mm]att.south);
  \end{tikzpicture}
  \caption{Claim lifecycle. Solid arrows are ledger appends, dashed arrows are computed from the index. A non-terminal claim becomes stale on binding drift and cannot be accepted until re-attested. Refuted and superseded are terminal.}
  \label{fig:lifecycle}
\end{figure}
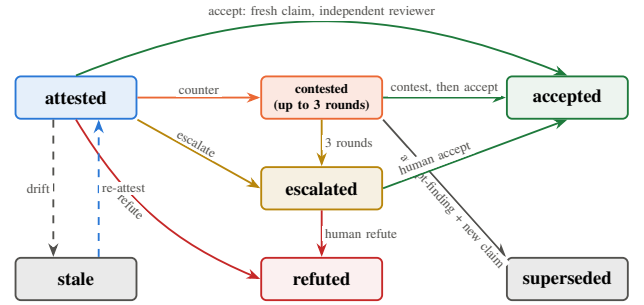

\subsection{Invariants in code}

Following Liza's stance that invariants belong in code rather than prompts, each of the following is enforced mechanically, most at ledger-append time and the rest at the gate. An attestor cannot review their own claim. A stale claim cannot be accepted. No verdict lands on a refuted or superseded claim. Only a human resolves an escalation. A contest must name a harm and a response must come from the attestor. A contested claim needs a response before the next verdict. Every ledger line carries an HMAC signature under a key held outside the repository, and an altered or unsigned line fails the gate. The gate passes only with fresh, independently accepted, plausible evidence for every requirement.

\subsection{The gate}

The gate is a function of the index, the ledger, and the changed modules. It computes the obligation, then for each required (predicate, module) pair finds the best available claim and classifies it as ok, weak, unreviewed, stale, escalated, failed, refuted, or missing. It passes iff every pair is ok and every signature verifies. Three plausibility checks stand between an accepted claim and ok, all floors rather than proofs. A \emph{tests-pass} command must look like a test runner that can test the subject's language. It must have produced output. And the number of tests that ran must not drop below what earlier evidence on an overlapping subject reported, unless a fresh human-approved claim covers the subject. We call the last check evidence monotonicity. It is the mechanical counterpart of the contract clause ``do not delete the failing test'', and it costs one integer per claim.

The gate guarantees that a merge carries fresh, signed, independently reviewed evidence for the whole blast radius, and that the evidence is at least plausible. It does not guarantee the evidence is sufficient. A reviewer can accept a weak suite, two models from one provider can share a blind spot, and a doer that can read the signing key can forge the ledger. These are the boundaries Liza states for its own verdicts~\citep{liza}, restated here because a mechanical floor invites over-reading.

\section{Implementation}

\subsection{Index}

\texttt{assay index} lists files through \texttt{git ls-files} so ignore rules apply, hashes and parses each file, resolves imports, groups files into modules, computes content and cone hashes, mines 90 days of history for activity, computes risk, detects per-language test, build, and typecheck commands, and writes \texttt{assay.json}. Python is parsed with the standard library AST. TypeScript and JavaScript, Go, Rust, Java and Kotlin, Ruby, and C use conservative regular expressions that recover import specifiers and exported names. Resolution understands relative and package imports, \texttt{go.mod} module paths, Rust crate paths and workspace crates, Java packages, and TypeScript path aliases. A tree-sitter backend can replace any parser behind the same interface. Zero dependencies keep the gate runnable anywhere Python runs.

A per-file cache keyed by size and modification time stores each file's hash and parse result, so a rebuild re-parses only files whose stat changed before recomputing the graph, the cone hashes, and the activity summary. This is Stacklit's Merkle-skip idea at file granularity. The pre-commit hook installed by \texttt{assay init -{}-hook} runs it and stages the index, so the committed index is always current.

\subsection{Brief, ledger, and interfaces}

The brief is what an agent reads before opening a source file. It carries the project's shape (hubs, hot spots, evidence commands per language), the change under way (changed modules, radius, obligation, required evidence), the state of the evidence (stale claims, open findings, the gate verdict), and a module directory ordered by fan-in and trimmed to a token budget, 600 by default. Doer and reviewer start from the same brief. Fig.~\ref{lst:brief} shows the brief for FastAPI with a hypothetical edit to its root package.

The ledger is a JSON-lines file under \texttt{.assay/}, committed as the audit trail. Every line is signed with HMAC-SHA256 under a key read from the environment or a git-ignored file. In a supervised deployment the process that executes the agent's tool calls holds the key and the agent never sees it. Appends take an exclusive lock for a read-validate-write cycle. Entries are never modified. Superseding a claim appends a new one that names the old.

The CLI exposes \texttt{init}, \texttt{index}, \texttt{brief}, \texttt{blast}, \texttt{attest}, \texttt{review}, \texttt{respond}, \texttt{gate}, \texttt{verify}, and \texttt{mcp}. The MCP server speaks JSON-RPC over stdio and exposes the same operations as tools, so Claude Code, Cursor, or any MCP client can use \assay{} without shelling out. In continuous integration, \texttt{assay gate -{}-base origin/main} exits non-zero when a change lacks fresh reviewed evidence for its radius. The implementation is \LocCore{} lines of Python with no runtime dependencies and \NumTests{} tests.

\begin{figure*}[t]
\begin{lstlisting}[basicstyle=\ttfamily\tiny]
# Assay brief: fastapi @ d5aba4195b | 195 modules, 1142 files, 113579 lines | python 1138, javascript 4
Commands: python: test pytest, typecheck mypy .; javascript: no test runner | Frameworks: fastapi, flask, pytest
Hubs (touch carefully): fastapi (fan-in 171, risk 1.00); tests (fan-in 41, risk 0.95); fastapi/openapi (fan-in 6, risk ..
Hot 90d: fastapi (38 commits); tests (18 commits); scripts (8 commits); fastapi/dependencies (7 commits); fastapi/open ..
Changed: fastapi -> blast radius 177/195 modules (91%), max risk 1.00 -> level 3 (escalate)
Dependents in radius: docs_src/additional_responses, docs_src/additional_status_codes, docs_src/advanced_middleware, d ..
Required evidence: tests-pass[docs_src/additional_responses, docs_src/additional_status_codes, docs_src/advanced_middl ..
Review: independent reviewer, adversarial, human sign-off, max 3 rounds
Modules (by fan-in):
- fastapi: FastAPI framework, high performance, easy to learn, fast to code, read | exports: FastAPI, BackgroundTasks, ..
- docs_src/additional_responses: additional_responses | exports: Item, Message, read_item(), Item, read_item(), Item,  ..
- docs_src/additional_status_codes: additional_status_codes | exports: upsert_item(), upsert_item() | deps: fastapi |  ..
- .. 192 more modules: `assay module <path>` or the MCP tool get_module
\end{lstlisting}
\caption{The brief for FastAPI with a hypothetical edit to the \texttt{fastapi} package (first lines, truncated at the right margin).}
\label{lst:brief}
\end{figure*}

\section{Evaluation}

Five questions, each answered by a scripted experiment with no model in the loop. E1: how many tokens does orientation cost with a brief, an index, exploration, or a full dump? E2: how much does Merkle skipping save on rebuilds? E3: what should a claim be bound to? E4: does the gate block the ways a doer can cheat and admit the honest paths? E5: how large is a typical change, and how is obligation distributed?

\subsection{Setup}

Five public repositories spanning five languages, chosen to overlap with the comparison set published by Stacklit: Express (JavaScript), Flask (Python), Gin (Go), Axum (Rust), and FastAPI (Python), each a shallow clone with 90 days of history, plus \assay{} itself. Table~\ref{tab:e1} gives sizes. Tokens are counted with \EstimatorName. Every figure and number below is produced by the released experiment scripts and injected into this document by \texttt{paper/make\_numbers.py}. A dogfooding run on a private TypeScript Electron application took one real change through attestation, independent review, and a passing gate with the project's own typecheck and build commands, and surfaced two obligation defects fixed in the released version.

\subsection{E1: tokens to orient}

Four ways for an agent to learn where things live: a full dump of every source file, an exploration proxy consisting of a directory listing plus the ten most-imported files (a conservative stand-in for the 8 to 12 files an agent reads before acting), the modules section of \texttt{assay.json}, and the brief at its default budget.

\begin{table*}[t]
  \centering\footnotesize
  \caption{E1: tokens required to orient an agent, by method. Index is the complete \texttt{assay.json} including per-file entries. Modules only is the section an agent would read whole. Gain is exploration divided by brief.}
  \label{tab:e1}
  \begin{tabular}{lrrrrrrrr}
\toprule
Repository & Lines & Modules & Full dump & Exploration & Index & Modules only & Brief & Gain \\
\midrule
express & 21,492 & 35 & 137,093 & 17,782 & 12,447 & 6,230 & 635 & 28$\times$ \\
flask & 18,345 & 16 & 134,439 & 46,673 & 11,330 & 3,882 & 595 & 78$\times$ \\
gin & 24,192 & 7 & 192,995 & 8,328 & 10,862 & 2,379 & 592 & 14$\times$ \\
axum & 46,686 & 90 & 335,509 & 45,581 & 42,112 & 19,308 & 621 & 73$\times$ \\
fastapi & 113,579 & 195 & 733,647 & 68,215 & 142,463 & 51,842 & 599 & 114$\times$ \\
assay & 4,668 & 5 & 58,582 & 24,430 & 4,487 & 1,626 & 434 & 56$\times$ \\
\bottomrule
\end{tabular}

\end{table*}

The brief costs \EOneMinBrief{} to \EOneMaxBrief{} tokens on the external corpora, which is \EOneMinXExpl$\times$ to \EOneMaxXExpl$\times$ less than the exploration proxy and \EOneMinXFull$\times$ to \EOneMaxXFull$\times$ less than a dump. The brief's cost is flat in repository size by construction. The modules section grows with the number of modules (FastAPI's \EOneFastapiModules{} modules, most of them documentation examples, cost \EOneFastapiIndexModulesOnly{} tokens) and is meant to be queried through \texttt{get\_module} rather than read whole. The ratios match those Stacklit reports for its own index~\citep{stacklit}. What no map carries is the radius, the obligation, and the stale claims.

\subsection{E2: incremental indexing}

\begin{table}[t]
  \centering\footnotesize
  \caption{E2: build time in milliseconds for a cold build, a warm build after one edited file, and a warm build with no change. Gain is cold divided by warm.}
  \label{tab:e2}
  \resizebox{\linewidth}{!}{\begin{tabular}{lrrrrrr}
\toprule
Repository & Files & Cold & One edit & No change & Gain (edit) & Gain (no change) \\
\midrule
express & 141 & 45 & 24 & 24 & 1.9$\times$ & 1.9$\times$ \\
flask & 83 & 94 & 32 & 23 & 3.0$\times$ & 4.1$\times$ \\
gin & 98 & 35 & 20 & 20 & 1.7$\times$ & 1.7$\times$ \\
axum & 305 & 62 & 38 & 34 & 1.6$\times$ & 1.8$\times$ \\
fastapi & 1142 & 508 & 106 & 93 & 4.8$\times$ & 5.4$\times$ \\
assay & 31 & 56 & 31 & 22 & 1.8$\times$ & 2.5$\times$ \\
\bottomrule
\end{tabular}
}
\end{table}

A warm rebuild with nothing changed is \ETwoMinSpeedupNoop$\times$ to \ETwoMaxSpeedupNoop$\times$ faster than a cold build (Table~\ref{tab:e2}), and with one edited file it re-parses exactly that file. On FastAPI (\ETwoFastapiFiles{} files) a cold build takes \ETwoFastapiCold{} ms and a one-file rebuild \ETwoFastapiWarmEdit{} ms. The residual cost is the history scan and graph recomputation, both cacheable. For a pre-commit hook this is below the noise floor of the commit itself.

\subsection{E3: what to bind a claim to}

One claim per module. For $k \in \{1, 2, 4, 8, 16, 32\}$ we draw \EThreeTrials{} random sets of $k$ edited modules and compute, under three binding strategies, the fraction of claims that must be re-verified (cost) and the fraction that dependency semantics say must be re-verified but the strategy leaves fresh (missed). Repo hash binds to the root hash, self hash to the module's own content hash, and cone hash is \assay{}'s default.

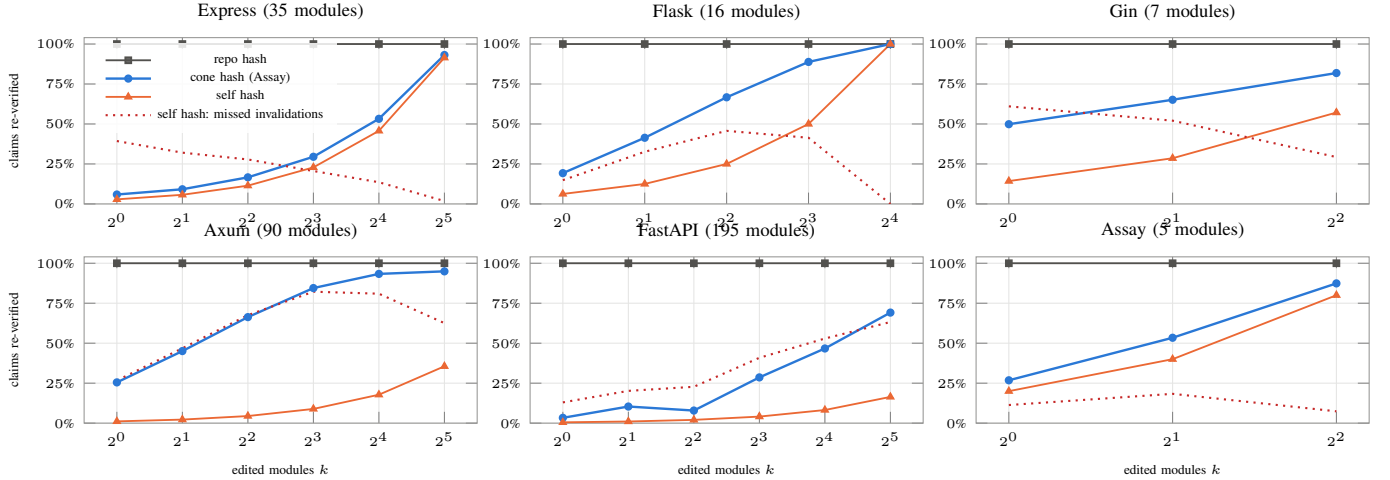
\begin{figure*}[t]
  \centering
  \begin{tikzpicture}
    \begin{groupplot}[
      group style={group size=3 by 2, horizontal sep=7mm, vertical sep=7mm,
                   xlabels at=edge bottom, ylabels at=edge left},
      width=5.2cm, height=2.2cm, scale only axis,
      xmode=log, log basis x=2, xtick=data,
      ymin=0, ymax=1.04, ytick={0,0.25,0.5,0.75,1},
      yticklabel={\pgfmathparse{\tick*100}\pgfmathprintnumber{\pgfmathresult}\%},
      tick label style={font=\tiny}, label style={font=\tiny}, title style={font=\scriptsize, yshift=-1mm},
      xlabel={edited modules $k$}, ylabel={claims re-verified},
      grid=major, grid style={assaymuted!15},
      axis line style={assaymuted!60},
      legend style={font=\tiny, draw=none, fill=white, fill opacity=0.85, text opacity=1, at={(0.03,0.97)}, anchor=north west, row sep=-2pt},
      cycle list={
        {assaymuted, mark=square*, mark size=1.1pt, line width=0.7pt},
        {assayblue, mark=*, mark size=1.1pt, line width=0.9pt},
        {assayorange, mark=triangle*, mark size=1.3pt, line width=0.7pt},
        {assaybad, dotted, mark=none, line width=0.8pt}},
    ]
    \nextgroupplot[title={Express (\EOneExpressModules{} modules)}]
      \addplot table[x=k, y=repo] {tables/e3_express.dat}; \addlegendentry{repo hash}
      \addplot table[x=k, y=cone] {tables/e3_express.dat}; \addlegendentry{cone hash (Assay)}
      \addplot table[x=k, y=self] {tables/e3_express.dat}; \addlegendentry{self hash}
      \addplot table[x=k, y=selfmissed] {tables/e3_express.dat}; \addlegendentry{self hash: missed invalidations}
    \nextgroupplot[title={Flask (\EOneFlaskModules{} modules)}]
      \addplot table[x=k, y=repo] {tables/e3_flask.dat};
      \addplot table[x=k, y=cone] {tables/e3_flask.dat};
      \addplot table[x=k, y=self] {tables/e3_flask.dat};
      \addplot table[x=k, y=selfmissed] {tables/e3_flask.dat};
    \nextgroupplot[title={Gin (\EOneGinModules{} modules)}]
      \addplot table[x=k, y=repo] {tables/e3_gin.dat};
      \addplot table[x=k, y=cone] {tables/e3_gin.dat};
      \addplot table[x=k, y=self] {tables/e3_gin.dat};
      \addplot table[x=k, y=selfmissed] {tables/e3_gin.dat};
    \nextgroupplot[title={Axum (\EOneAxumModules{} modules)}]
      \addplot table[x=k, y=repo] {tables/e3_axum.dat};
      \addplot table[x=k, y=cone] {tables/e3_axum.dat};
      \addplot table[x=k, y=self] {tables/e3_axum.dat};
      \addplot table[x=k, y=selfmissed] {tables/e3_axum.dat};
    \nextgroupplot[title={FastAPI (\EOneFastapiModules{} modules)}]
      \addplot table[x=k, y=repo] {tables/e3_fastapi.dat};
      \addplot table[x=k, y=cone] {tables/e3_fastapi.dat};
      \addplot table[x=k, y=self] {tables/e3_fastapi.dat};
      \addplot table[x=k, y=selfmissed] {tables/e3_fastapi.dat};
    \nextgroupplot[title={Assay (\EOneAssayModules{} modules)}]
      \addplot table[x=k, y=repo] {tables/e3_assay.dat};
      \addplot table[x=k, y=cone] {tables/e3_assay.dat};
      \addplot table[x=k, y=self] {tables/e3_assay.dat};
      \addplot table[x=k, y=selfmissed] {tables/e3_assay.dat};
    \end{groupplot}
  \end{tikzpicture}
  \caption{E3: fraction of claims re-verified per random edit of $k$ modules (\EThreeTrials{} trials per point) under three binding strategies. Repo-hash binding re-verifies everything. Self-hash binding is cheapest but leaves the dotted fraction of required invalidations undetected. Cone binding is the sound minimum.}
  \label{fig:e3}
\end{figure*}

\begin{table}[t]
  \centering\footnotesize
  \caption{E3 at $k = 4$ edited modules: re-verification cost per strategy and the fraction of required invalidations the self-hash strategy misses.}
  \label{tab:e3}
  \resizebox{\linewidth}{!}{\begin{tabular}{lrrrrr}
\toprule
Repository & Modules & Repo hash & Cone hash & Self hash & Self hash misses \\
\midrule
express & 35 & 100\% & 16.6\% & 11.4\% & 27.8\% \\
flask & 16 & 100\% & 66.7\% & 25.0\% & 45.7\% \\
gin & 7 & 100\% & 81.9\% & 57.1\% & 29.3\% \\
axum & 90 & 100\% & 66.3\% & 4.4\% & 67.6\% \\
fastapi & 195 & 100\% & 7.9\% & 2.1\% & 22.7\% \\
assay & 5 & 100\% & 87.4\% & 80.0\% & 7.4\% \\
\bottomrule
\end{tabular}
}
\end{table}

The result is what Propositions~\ref{prop:sound} and~\ref{prop:minimal} predict, now with magnitudes (Fig.~\ref{fig:e3}, Table~\ref{tab:e3}). At four edited modules, cone binding re-verifies between \EThreeMinConeCost{} (FastAPI) and \EThreeMaxConeCost{} (Gin) of claims where repo binding re-verifies all of them, and self binding misses between \EThreeMinSelfMissed{} and \EThreeMaxSelfMissed{} of the claims that must be re-verified. The spread is structural. FastAPI's graph is wide and shallow, with hundreds of example modules depending on one package, so a random edit rarely lands in a cone. Axum's workspace crates form a deep chain, so most edits do. Neither the cheap-but-unsound nor the sound-but-total strategy is acceptable for a gate. Cone binding is the only one that is both, and its cost is a property of the repository, not the tool.

\subsection{E4: the gate against an adversarial doer}

A twelve-file Python project with a real pytest suite. Each scenario starts from a clean copy. A scripted doer edits and makes claims, honestly or not, a scripted reviewer acts, and the gate decides. Table~\ref{tab:e4} lists the scenarios and the mechanism that fired.

\begin{table*}[t]
  \centering\footnotesize
  \caption{E4: scripted doer behaviours and the gate's decision. A rejected ledger entry is treated as absent, so the forged scenario reports only the signature mechanism.}
  \label{tab:e4}
  \begin{tabular}{p{7.2cm}lp{4.6cm}}
\toprule
Scenario & Gate & Mechanisms that fired \\
\midrule
honest leaf edit (L1): tests on the radius, independent accept & PASS & none \\
honest hub edit (L3): + lint, secrets scan, human sign-off & PASS & none \\
no evidence at all & FAIL & coverage \\
evidence covers the edited module, not its dependents & FAIL & coverage \\
evidence produced before the edit (stale) & FAIL & freshness \\
doer approves own claim & FAIL & separation \\
forged ledger entries (claim + verdict) & FAIL & signature \\
code broken, real test run recorded & FAIL & exit code \\
code broken, fake test command & FAIL & plausibility \\
code broken, failing test deleted & FAIL & monotonicity \\
reviewer refutes the claim & FAIL & review \\
three rounds, escalated, human resolves & PASS & review \\
\bottomrule
\end{tabular}

\end{table*}

\EFourAllExpected{} \EFourScenarios{} scenarios end where the design says they should. The \EFourPassed{} honest paths pass: a leaf edit at level 1, a hub edit at level 3 with lint, secrets scan, and human sign-off, and an escalation resolved by a human, which was correctly blocked while escalated. The \EFourBlocked{} adversarial ones fail, each on the mechanism built for it. Two rows deserve comment. \emph{Failing test deleted} is the reward-tampering behaviour documented in the model literature~\citep{denison2024subterfuge}. The doer breaks the code, replaces the test that would catch it with a placeholder, runs a real test command that genuinely passes, and gets an accept from a reviewer who did not look. The gate still fails, because the number of tests that ran dropped from the earlier accepted evidence and no human signed off on the removal. \emph{Fake test command} passes a real reviewer for the same reason and fails on plausibility. The point of E4 is not that these tricks are hard to think of, but that the floor holds without anyone thinking of them at review time.

\subsection{E5: how large is a change}

\begin{table}[t]
  \centering\footnotesize
  \caption{E5: single-module blast radius as a fraction of all modules, and the number of modules whose edit falls at each obligation level.}
  \label{tab:e5}
  \resizebox{\linewidth}{!}{\begin{tabular}{lrrrrrrr}
\toprule
Repository & Modules & Median radius & p90 radius & L0 & L1 & L2 & L3 \\
\midrule
express & 35 & 5.7\% & 8.6\% & 13 & 22 & 0 & 0 \\
flask & 16 & 6.2\% & 87.5\% & 13 & 0 & 0 & 3 \\
gin & 7 & 57.1\% & 71.4\% & 1 & 1 & 1 & 4 \\
axum & 90 & 1.1\% & 91.1\% & 66 & 0 & 0 & 24 \\
fastapi & 195 & 0.5\% & 1.0\% & 149 & 38 & 0 & 8 \\
assay & 5 & 20.0\% & 20.0\% & 4 & 0 & 1 & 0 \\
\bottomrule
\end{tabular}
}
\end{table}

Across the external corpora, between \EFiveMinSingle{} and \EFiveMaxSingle{} of modules reach only themselves (Table~\ref{tab:e5}). The median single-module radius is \EFiveFastapiMedian{} of modules for FastAPI and \EFiveAxumMedian{} for Axum, while the 90th percentile in Axum is \EFiveAxumPninety{}, because a handful of hubs (\EFiveAxumTopModule{}, reached by \EFiveAxumTopBr{} modules) sit under everything. The obligation policy follows that shape. Most edits are light or standard and the hubs escalate. A uniform review policy spends the same effort on a documentation example as on the core package. The obligation spends it where the radius is.

\subsection{Threats to validity}

Import graphs from regular expressions under-approximate dependencies (dynamic imports, reflection, generated code), which makes cone hashes under-approximate and would weaken Proposition~\ref{prop:sound} in practice. Authoritative edge sources and a tree-sitter backend are the remedy, and edge recall should be reported per repository. Module granularity is a directory, and a large directory inflates radii. All conditions share one tokenizer. The exploration proxy stands in for agent behaviour rather than measuring it. E4 measures a floor against scripted behaviours designed by the same people who designed the checks, so independent red-teaming is warranted. Risk and the obligation thresholds are policy.

\section{Related Work}

\textbf{Context for agents.} Repository maps and packers~\citep{aider,repomix,stacklit}, code intelligence indexes such as SCIP~\citep{scip}, and structure-aware retrieval in AutoCodeRover~\citep{zhang2024autocoderover} and the deliberately minimal Agentless~\citep{xia2024agentless} all reduce the tokens an agent spends locating code. Long-context degradation~\citep{liu2024lost} motivates compactness beyond cost. None bind verification to structure.

\textbf{Agent frameworks and accountability.} Multi-agent frameworks~\citep{hong2024metagpt,qian2024chatdev,wu2023autogen,crewai} and agent-computer interfaces such as SWE-agent~\citep{yang2024sweagent} and OpenHands~\citep{wang2024openhands} shape what agents can do. Failure taxonomies~\citep{cemri2025mast} and the reward-tampering literature~\citep{sharma2024sycophancy,denison2024subterfuge} explain why outcomes need verification the agent cannot game. Liza~\citep{liza} is the closest system in spirit and the source of our protocol. Our contribution relative to it is the binding of verdicts to content, which turns ``re-review after rebase'' from procedure into arithmetic. Long-horizon measurements~\citep{kwa2025longtasks} suggest unattended runs will lengthen, raising the value of a floor that needs no attention.

\textbf{Content addressing, build systems, provenance.} Merkle trees~\citep{merkle1987} underlie git, Nix~\citep{dolstra2004nix}, and Bazel's action cache~\citep{bazel}. Mokhov et al.~\citep{mokhov2018build} give the vocabulary of verifying traces that our ledger instantiates for assertions. in-toto~\citep{intoto2019} and SLSA~\citep{slsa} attest to what ran over which inputs in a supply chain. A claim with its bindings, evidence, and signature is structurally an in-toto link moved into the inner loop, with the dependency cone as the input set. Proposition~\ref{prop:radius} is also a regression test selection oracle in the manner of Ekstazi~\citep{gligoric2015ekstazi}, with the dependency graph shared with the agent's context.

\textbf{Argumentation.} The counter and contest exchange is a Toulmin argument~\citep{toulmin1958}: evidence is the grounds, the cone binding is the warrant for this code state, a counter is a rebuttal, and when the binding drifts the warrant lapses.

\section{Limitations and Future Work}

\textbf{Parsers.} Regular-expression extraction is the largest source of unsoundness. A tree-sitter backend and, where available, authoritative edge sources (\texttt{go list}, \texttt{tsc -{}-listFiles}, \texttt{cargo metadata}) should replace it, and edge recall should be a reported property of every index.

\textbf{Granularity.} Directory modules are coarse. File-level cones are a straightforward extension. Symbol-level cones need a real call graph.

\textbf{Evidence adequacy.} The gate checks plausibility, not sufficiency. Coverage deltas on changed lines and per-test identities are the next mechanical checks, and per-test identities would make monotonicity exact.

\textbf{Key placement.} The floor depends on the signing key being outside the agent's reach. We describe the deployment pattern but do not enforce it.

\textbf{Beyond tests.} The same binding covers documentation (a \emph{documents} claim goes stale when the code it describes moves), lockfiles (\emph{no-known-vulnerabilities} bound to the lockfile hash), infrastructure plans, and monorepo federation through a dependency's version hash as a leaf. A pull request organized by claims and radius would show a reviewer what was verified, what went stale, and what is still owed.

\textbf{Models in the loop.} The study this paper does not contain measures resolved rate, tokens to first correct patch, tokens to gate pass, and the rate at which the gate blocks patches that pass a harness but fail hidden tests, across arms with and without the brief and the gate, on a benchmark such as SWE-bench Verified~\citep{jimenez2024swebench}. We have designed it and will preregister it.

\section{Conclusion}

Context and accountability for coding agents are the same problem seen at two times. Keying claims by the Merkle hash of the dependency cone of the code they cover makes staleness a hash comparison, makes blast radius and invalidation the same set, and lets a merge gate that consults no model be both precise and cheap. \assay{} is a small, dependency-free implementation of that idea, and the experiments show the floor holding on real repositories and against scripted attacks. The best summary is also the design rule: claims should decay with the code.

\section*{Acknowledgment}

The protocol and the ``mechanical, not instructional'' stance come from Liza. The committed, Merkle-skipped index and its MCP surface come from Stacklit. Assay shares no code with either project and would not exist without both. An AI coding assistant helped produce the figures and tables. The authors verified every claim and number and are responsible for the content.

\bibliographystyle{IEEEtranN}
\bibliography{refs}

\end{document}